\documentclass{patmorin}
\usepackage{pat,float,subcaption}
\usepackage[T1]{fontenc}
\usepackage[utf8]{inputenc}
\usepackage{amssymb,mathtools}
\usepackage{graphicx}
\usepackage{microtype}
\usepackage[section]{placeins}
\usepackage[longnamesfirst,numbers,sort&compress]{natbib}
\usepackage[inline]{enumitem}

\newcommand{\CH}{\operatorname{CH}}
\newcommand{\diam}{\operatorname{diam}}

\newcommand{\len}{\ell}
\newcommand{\kappaGT}{\kappa}
\newcommand{\T}{T}
\newcommand{\ch}{\CH}

\title{On the Spanning Ratio of the Greedy Triangulation\\for Convex Point Sets}
\author{Prosenjit Bose\thanks{School of Computer Science, Carleton University,
Ottawa, Ontario, Canada.
Emails: \email{jit@scs.carleton.ca},
\email{bobby.miraftab@gmail.com},
\email{anil@scs.carleton.ca}, and
\email{michiel.smid@gmail.com}.} \qquad \and Jean-Lou de Carufel\thanks{School of Electrical Engineering and Computer Science,
University of Ottawa, Ottawa, Ontario, Canada.
Emails: \email{jdecaruf@uottawa.ca} and
\email{leotheocharous3@gmail.com}.} \qquad \and Anil Maheshwari\footnotemark[1] \qquad \and Bobby Miraftab\footnotemark[1] \qquad \and Michiel Smid\footnotemark[1] \qquad \and Leonidas Theocharous\footnotemark[2]}
\hypersetup{
  pdftitle={On the Spanning Ratio of the Greedy Triangulation for Convex Point Sets},
  pdfauthor={Prosenjit Bose, Jean-Lou de Carufel, Anil Maheshwari, Bobby Miraftab, Michiel Smid, Leonidas Theocharous},
  pdfsubject={Computational geometry; plane geometric spanners; greedy triangulations},
  pdfkeywords={greedy triangulation, geometric spanner, spanning ratio, convex position}
}

\date{}

\begin{document}
\maketitle

\begin{abstract}
The greedy triangulation of a finite planar point set is obtained by considering all
segments in nondecreasing order of length and inserting each segment that does not cross
an earlier one. Its spanning ratio is known to be bounded by a universal constant, but
the standard bound obtained from the diamond and good-polygon properties is about
$11739.1$. We prove a substantially smaller bound for points in convex position. In particular,
for every finite point set $P\subset\R^2$ in convex position and every pair
$u,v\in P$, the greedy triangulation contains a $u$--$v$ path of length at most
$\kappaGT |uv|$, where $\kappaGT<17.814$. Thus, the greedy triangulation of a
convex point set is an $18$-spanner.
\end{abstract}

\section{Introduction}

A geometric graph on a finite point set $P\subset\R^2$ is a \emph{$t$-spanner} if every
pair of vertices is joined by a path whose length is at most $t$ times their Euclidean
distance. The least admissible value of $t$ is the graph's \emph{spanning ratio}, also
called its stretch factor or dilation. Geometric spanners provide sparse networks that
approximately preserve the metric of the complete Euclidean graph and are a central
object in computational geometry; see the survey of Bose and Smid~\cite{BoseSmid2013}.

The \emph{greedy triangulation} processes all segments determined by $P$ in
nondecreasing order of length and retains a segment exactly when it does not cross a
previously retained segment. The resulting graph is plane and maximal, hence a
triangulation. Das and Joseph~\cite{DasJoseph1989} proved that greedy triangulations have
bounded spanning ratio. Greedy triangulations have also been studied as approximations
to minimum-weight triangulations and from an algorithmic perspective
\cite{LevcopoulosLingas1987,LevcopoulosLingas1992}.

For completeness, we recall the two local properties underlying the general spanning-ratio
bound. Fix $0<\alpha<\pi/2$ and an edge $xy$. On the two sides of the line through $xy$,
let $D^+_\alpha(xy)$ and $D^-_\alpha(xy)$ be the open isosceles triangles having $xy$ as
their common base and base angles $\alpha$. The edge $xy$ has the
\emph{$\alpha$-diamond property} if at least one of these two triangles contains no point
of $P$; a plane geometric graph has the $\alpha$-diamond property if each of its edges
does. Thus every graph edge is incident to an empty angular region of a fixed shape.

The second condition controls detours around faces. Two vertices $a,b$ on the boundary of
a face $f$ are \emph{visible through $f$} when the segment $ab$ is contained in the
closure of $f$. A plane geometric graph has the \emph{$d$-good-polygon property} if, for
every face $f$ and every such visible pair $a,b$, the shorter of the two boundary paths
from $a$ to $b$ has length at most $d|ab|$. Every triangulation has the
$1$-good-polygon property: in a bounded triangular face a visible pair is joined by a
side, and on the outer face the visible pairs are adjacent hull vertices. The diamond
property supplies empty regions that constrain how a route can progress, while the
good-polygon property bounds the cost of replacing portions of that route by walks along
face boundaries. These notions were introduced by Das and Joseph and sharpened in the
generalized diamond-spanner framework of Bose, Lee, and Smid
\cite{DasJoseph1989,BoseLeeSmid2007}.

The greedy triangulation has the $\pi/6$-diamond property, and, being a triangulation, it
has the $1$-good-polygon property. The estimate of Bose, Lee, and
Smid~\cite{BoseLeeSmid2007} therefore gives, for arbitrary planar point sets,
\[
  \frac{8d(\pi-\alpha)^2}{\alpha^2\sin^2(\alpha/4)}
  =\frac{8(\pi-\pi/6)^2}{(\pi/6)^2\sin^2(\pi/24)}
  \approx 11739.1.
\]
This bound is evidently far from the behavior suggested by known examples. The best known lower bound is given by a configuration of six points in convex position, for which the greedy triangulation has a spanning ratio of $2.0268$ \cite{DumitrescuGhosh2016}.

A sharper undirected bound also has consequences for oriented geometric
spanners. Buchin et al.~\cite{BuchinEtAl2026} showed that a consistent orientation
of the greedy triangulation has oriented dilation at most $7.2t_g$, where $t_g$
is an upper bound on its undirected spanning ratio. Thus, every improvement in
the undirected bound for convex point sets immediately yields a corresponding
improvement in the oriented setting.

In this paper, we study the greedy triangulation for convex point sets and obtain the following result.

\begin{thm}[Spanning ratio]\label{thm:main}
Let $P\subset\R^2$ be a finite point set in convex position, and let $T$ be its greedy
triangulation. Then, for all $u,v\in P$,
\[
  d_T(u,v)\le \kappaGT |uv|,
  \qquad
  \kappaGT
  =6+\frac{72}{17}+\sqrt{2}+2\sqrt{9+\frac{144}{289}}
  <17.814.
\]
In particular, the greedy triangulation of $P$ is an $18$-spanner.
\end{thm}

The proof has two stages. First, unless one of the two hull paths between the chosen
vertices $u$ and $v$ is already short, we show that the portion of the convex hull in a
carefully chosen neighborhood of $uv$ consists of two $x$-monotone convex chains, one
containing $u$ and the other containing $v$. Second, we prove that the greedy
triangulation contains an edge joining these two chains inside that neighborhood; we call
such an edge a \emph{bridge}. Combining the bridge with suitable subpaths of the two hull
chains yields a $u$--$v$ path of the desired length. A detailed overview of the argument is
given in Section~\ref{sec:overview}.

Although Theorem 1 substantially improves the general upper bound in the convex-position setting, a considerable gap remains between our upper bound of 18 and the best known lower bound of 2.0268. Closing this gap is an interesting direction for future work.

\section{Preliminaries}\label{sec:prelim}

Throughout the paper, $P\subset\R^2$ is a finite point set in convex position. For
simplicity, assume that no three points of $P$ are collinear and that no two candidate
edges have the same length. Standard symbolic perturbation, or a fixed deterministic
tie-breaking rule, removes the latter assumption without changing the arguments.
For a point $p\in\R^2$, write $p=(x_p,y_p)$.

Let $\CH(P)$ denote the convex hull of $P$, and let $T=(P,E)$ be the greedy triangulation
of $P$. Thus the greedy algorithm considers all segments with endpoints in $P$ in
nondecreasing order of length and inserts a segment precisely when it does not cross any
previously inserted segment. Here and below, two segments \emph{cross} when their
relative interiors intersect. Every hull edge belongs to $T$.

For a vertex $v\in P$, let $\angle v$ denote the interior angle of $\CH(P)$ at $v$. For a
polygonal path $\pi=(v_0,\ldots,v_k)$, define
\[
  \len(\pi)=\sum_{i=1}^{k}|v_{i-1}v_i|.
\]
We write $d_T(u,v)$ for the shortest-path distance between $u$ and $v$ in $T$. The
spanning ratio of $T$ is
\[
  \max_{\{u,v\}\subseteq P}\frac{d_T(u,v)}{|uv|}.
\]
We use the same notation $uv$ for the segment joining $u$ and $v$ and, when that segment
belongs to a graph, for the corresponding edge.

We begin with two elementary geometric estimates.

\begin{obs}\label{obs:angle}\label{obs:triangle}
Let $a,b,x\in\R^2$, and suppose that the angle at $x$ in the triangle $axb$ is at least
$\varphi$, where $0<\varphi<\pi$. Then
\[
  |ax|+|xb|\le \frac{|ab|}{\sin(\varphi/2)}.
\]
\end{obs}

\begin{proof}
Set $A=|ax|$, $B=|xb|$, and let $\theta=\angle axb\ge\varphi$. The cosine rule gives
\[
 |ab|^2=(A-B)^2+4AB\sin^2(\theta/2)
 \ge (A+B)^2\sin^2(\theta/2).
\]
Since $\sin(\theta/2)\ge\sin(\varphi/2)$, the claim follows.
\end{proof}

\begin{obs}\label{obs:chain-triangle}\label{obs:convex-chain}
Let $\gamma$ be a convex polygonal chain from $p$ to $q$, and let $x$ be a point such that
$\gamma$ lies in the triangle $pxq$. Then
\[
  \len(\gamma)\le |px|+|xq|.
\]
\end{obs}

\begin{proof}
The chain $\gamma$ together with the chord $pq$ bounds a convex polygon contained in the
triangle $pxq$. Perimeter is monotone under inclusion for planar convex sets, so
$\len(\gamma)+|pq|\le |px|+|xq|+|pq|$.
\end{proof}

We also use the following standard property of the greedy triangulation.

\begin{obs}\label{obs:blocker}
Let $ab$ be a segment with endpoints in $P$. If $ab\notin E(\T)$, then there is an
edge $pq\in E(\T)$ such that $pq$ crosses $ab$ and $|pq|\le |ab|$.
We call such an edge a \emph{blocker} of $ab$.
\end{obs}

\begin{proof}
When the greedy algorithm considers the segment $ab$, it inserts $ab$ unless $ab$
crosses an already inserted edge. Thus, if $ab$ is not inserted, then at that moment
there is an already inserted edge $pq$ crossing $ab$. Since the greedy algorithm
processes edges in nondecreasing order of length, $|pq|\le |ab|$.
\end{proof}

\section{Proof Overview}\label{sec:overview}

Fix $u,v\in P$ and set $L=|uv|$.
If one of $u$ or $v$ has hull angle at most $2\pi/3$, Lemma~\ref{lem:bounded-angle} gives a hull path of length less than $4L$.
Thus the hard case is when both endpoint angles are large.

In that case we place two concentric rectangles around the midpoint of $uv$.
The larger rectangle $Q_2$ has dimensions $6L\times\frac{24}{17}L$, and the smaller rectangle $Q_1$ has the same height and dimensions $\frac85L\times\frac{24}{17}L$.
A preliminary reduction shows that, unless a hull path already has length at most
$\kappaGT L$, the portion of $\partial\ch(P)$ inside $Q_2$ consists of two
$x$-monotone convex chains, denoted $A$ and $B$, with $u\in A$ and $v\in B$.
Refer to Figure~\ref{fig:bridge-path}(i) for an illustration.

At this point the target becomes a \emph{bridge}: an edge of $\T$ with one endpoint on $A$ and one endpoint on $B$, both endpoints lying in $Q_2$.
This is enough because a bridge $ab$ immediately gives the path
\[
 u \leadsto_A a,\quad ab,\quad b\leadsto_B v.
\]
The two subpaths are contained in the controlled rectangle and the bridge has
length at most $\diam(Q_2)$; Lemma~\ref{lem:bridge-short-path} checks that the
resulting path has length at most $\kappaGT L$. This path is shown in
Figure~\ref{fig:bridge-path}(ii).

\begin{figure}
\begin{center}
\includegraphics{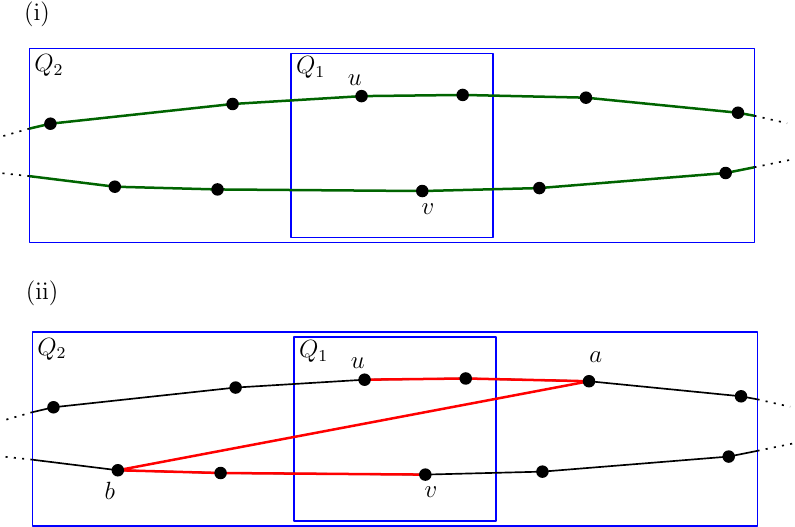}
\end{center}
\caption{(i) The rectangles $Q_1$ and $Q_2$ are drawn in blue, with $Q_1$ slightly displaced from $\partial Q_2$ for the purpose of illustration. We reduce to the case where $\partial CH(P) \cap Q_2$ consists of two $x$-monotone convex chains as shown: an upper chain $A$ with $u\in A$, and a lower chain $B$ with $v\in B$. Moreover, these two chains only cross the short sides of $Q_2$. (ii) The existence of a bridge in $E(\T)$ implies a short path between $u$ and $v$. Here, the path is drawn in red.}
 \label{fig:bridge-path}
\end{figure}

It remains only to force the bridge.
Assume no bridge exists.
Starting from the left side of $Q_2$, we sweep from left to right through segments joining $A$ to $B$.
Each such segment is absent from $\T$, so the blocker property (see Observation~\ref{obs:blocker}) supplies a greedy edge crossing it and no longer than it.
The sweep produces an edge of $\T$ entering $Q_1$ from the left of $Q_2$ and ending at a point $x$ of one chain inside $Q_1$.
A complementary sweep from the right produces an edge entering $Q_1$ from the right and ending at a point $y$ of the opposite chain. Refer to Figure~\ref{fig:contradiction-one} for an illustration of the two obtained edges.
The segment joining $x$ and $y$ lies in $Q_1$ and has length less than $\diam(Q_1)$.
A blocker of this short segment cannot stay inside $Q_2$ without either being the forbidden bridge or crossing one of the two already constructed edges as illustrated in Figure~\ref{fig:contradiction-two}(i). It also cannot leave $Q_2$, because then it is longer than $\diam(Q_1)$, as indicated in Figure~\ref{fig:contradiction-two}(ii).
This contradiction proves the existence of a bridge.

\begin{figure}
\begin{center}
\includegraphics{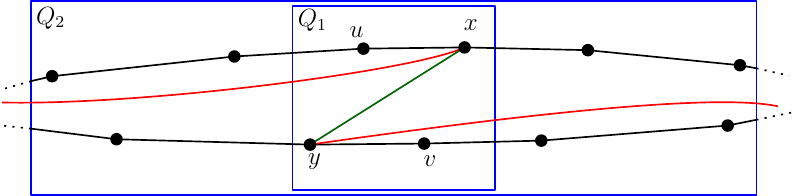}
\end{center}
\caption{Assuming that no bridge exists, we force the existence of the two shown red edges; they both have an endpoint outside $Q_2$, while the two remaining endpoints lie on opposite chains within $Q_1$.}
 \label{fig:contradiction-one}
\end{figure}

\begin{figure}
\begin{center}
\includegraphics{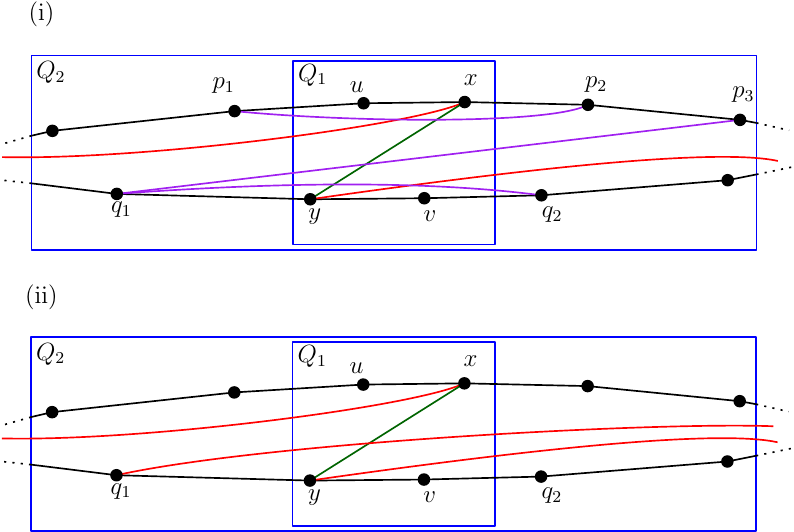}
\end{center}
\caption{If $xy$ is not in $E(\T)$, it must be blocked by a shorter edge. (i) If the edge that blocks $xy$ remains in $Q_2$, then it is either a bridge (this is the case for the edge $q_1p_3$ in the figure), and if not, it must cross one of the two red edges, violating planarity (this is the case for the edges $p_1p_2$ and $q_1q_2$ in the figure). (ii) If the edge that blocks $xy$ exits $Q_2$ (this is the case for the red edge starting at $q_1$), then we show that any such edge must be longer than $xy$, leading to a contradiction.}
 \label{fig:contradiction-two}
\end{figure}

\section{Bounding the Spanning Ratio}\label{sec:bound}

The following lemma handles the case in which one endpoint has relatively small
interior angle.

\begin{lem}\label{lem:bounded-angle}
Let $u\in P$ and suppose that $\angle u\le \theta<\pi$. Then, for every $v\in P$,
at least one of the two hull paths from $u$ to $v$ has length at most
\[
\frac{1}{\sin((\pi-\theta)/4)}\,|uv|.
\]
In particular,
\[
d_\T(u,v)\le \frac{1}{\sin((\pi-\theta)/4)}\,|uv|.
\]
\end{lem}

\begin{proof}
If $v$ is a neighbor of $u$ on $\ch(P)$, then $uv$ is a hull edge, hence an edge of
$\T$, and the claim is immediate.

Assume now that $v$ is not a neighbor of $u$. Let $u_1$ and $u_2$ be the two neighbors
of $u$ on $\ch(P)$. Since $P$ is in convex position, the whole point set lies in the
wedge bounded by the rays from $u$ through $u_1$ and $u_2$.

Let $\ell$ be a supporting line of $\ch(P)$ through $v$, chosen so that it intersects both boundary rays of the wedge at $u$ and is not parallel to $uu_1$ or $uu_2$. This line intersects the rays through
$uu_1$ and $uu_2$ at points $v_1$ and $v_2$, respectively. Then $u,v_1,v_2$ form a
triangle containing $v$ on the side $v_1v_2$, and the angle at $u$ in this triangle is
$\angle u$. Hence
\[
\angle uv_1v_2+\angle uv_2v_1=\pi-\angle u\ge \pi-\theta.
\]
Thus at least one of these two angles is at least $(\pi-\theta)/2$. Without loss of
generality, assume that $\angle uv_1v_2\ge(\pi-\theta)/2$. Since $v,v_1,v_2$ are
collinear, we also have $\angle uv_1v\ge(\pi-\theta)/2$.

Now consider the hull path from $u$ to $v$ that starts with the edge $uu_1$. This path
is a convex polygonal chain contained in the triangle $uv_1v$. By
Observation~\ref{obs:convex-chain}, its length is at most $|uv_1|+|v_1v|$. Applying
Observation~\ref{obs:triangle} to the triangle $uv_1v$, with
$\varphi=\angle uv_1v\ge(\pi-\theta)/2$, gives
\[
|uv_1|+|v_1v|
\le
\frac{1}{\sin((\pi-\theta)/4)}\,|uv|.
\]
Thus one hull path from $u$ to $v$ has length at most
$\frac{1}{\sin((\pi-\theta)/4)}|uv|$. Since every hull edge belongs to $\T$, this hull
path is a path in $\T$, proving the claim.
\end{proof}

\begin{figure}
\begin{center}
\includegraphics{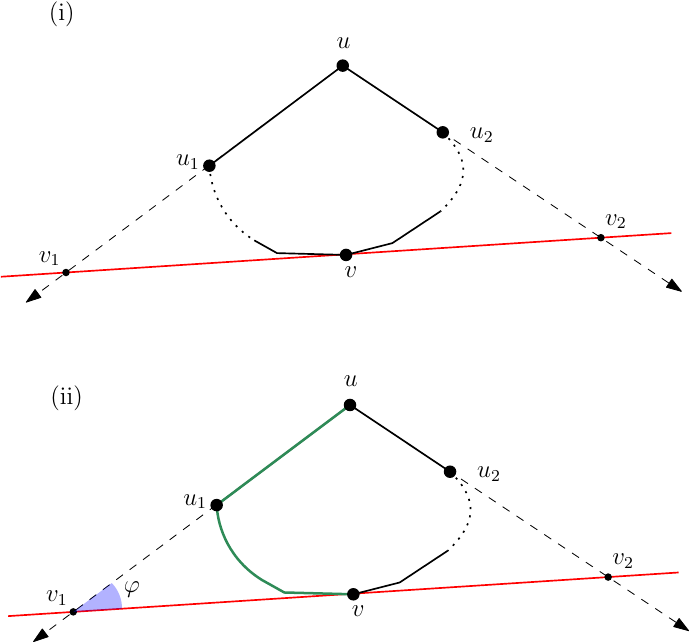}
\end{center}
\caption{ Illustration for the proof of Lemma~\ref{lem:bounded-angle}. (i) $\ch(P)$ is contained in the wedge bounded by the rays from $u$ through $u_1$ and $u_2$. The red line through $v$ supports $\ch(P)$ and intersects the two rays at the points $v_1$ and $v_2$. (ii) The green path follows the convex hull from $u$ to $v$ through $u_1$, and is contained in the triangle $uv_1v$. Therefore, it is shorter than $|uv_1|+|v_1v|$. The lower bound on the angle $\varphi$ combined with Observation~\ref{obs:triangle}, implies the lemma. }
\label{fig:bounded-angle}
\end{figure}

The following corollary is obtained by applying Lemma~\ref{lem:bounded-angle} with $\theta = 2\pi/3$.

\begin{cor}\label{cor:bounded-angle}
Let $u \in P$ and suppose that $\angle u \le 2\pi/3$. Then, for every $v \in P$,
\[
d_{\T}(u,v) \le \frac{1}{\sin(\pi/12)} |uv|
= (\sqrt2+\sqrt6)|uv|.
\]
In particular,
$d_{\T}(u,v) < 4 |uv|.$
\end{cor}

Next, we show that it suffices to consider a local configuration around the segment $uv$ in which the convex hull is contained in a narrow rectangle and exits this rectangle through both of its short sides. To formulate this reduction, we introduce the following definitions.

Let $H_1$ and $H_2$ be the two hull paths from $u$ to $v$. For a hull path $H$ from $u$ to $v$, let $\lambda_u(H)$ be the line containing the first edge of $H$ at $u$, and let $\lambda_v(H)$ be the line containing the last edge of $H$ at $v$. If these two lines intersect, let $x(H)=\lambda_u(H)\cap\lambda_v(H)$, and define $\phi(H)=\angle ux(H)v$. If $\lambda_u(H)$ and $\lambda_v(H)$ are parallel, we set $\phi(H)=0$.

\begin{obs}\label{obs:path-in-triangle}
Suppose that $u$ and $v$ are nonadjacent vertices of $\ch(P)$ and that
$\lambda_u(H)$ and $\lambda_v(H)$ intersect. Then exactly one of the two
hull paths from $u$ to $v$ is contained in the triangle $x(H)uv$.
\end{obs}

\begin{proof}
Refer to Figure~\ref{fig:rectangle-one} for an illustration. The lines $\lambda_u(H)$ and $\lambda_v(H)$ support $\ch(P)$ at $u$ and $v$,
respectively. Hence $\ch(P)$ is contained in the union of the closed wedge bounded by
these two lines that contains $\ch(P)$ and the triangle $x(H)uv$. Since the boundary
of $\ch(P)$ consists of exactly two hull paths from $u$ to $v$, one of them lies in
the triangle $x(H)uv$, while the other lies on the opposite side. This proves the
claim.
\end{proof}

\begin{figure}
\begin{center}
\includegraphics{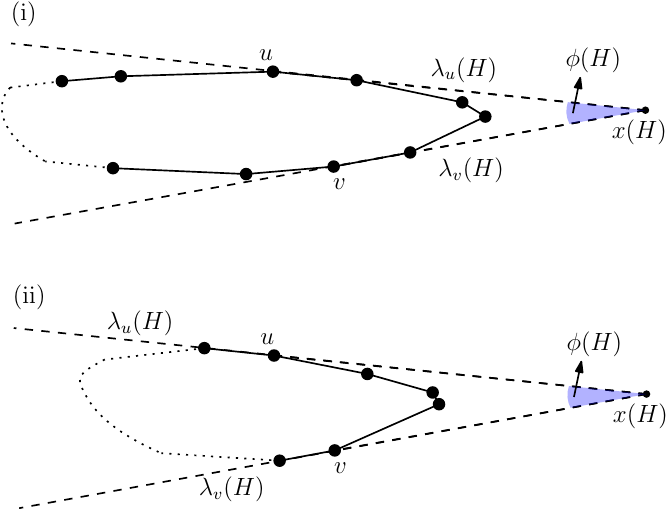}
\end{center}
\caption{ (i) Here, $H$ is the hull path from $u$ to $v$ in clockwise direction. It is contained in the triangle $x(H)uv$, and thus, it is shorter than $|x(H)u|+|x(H)v|$. (ii) Here, $H$ is the hull path from $u$ to $v$ in counterclockwise direction. This figure is used to illustrate that the hull path contained in the triangle $x(H)uv$ might not be $H$, but rather the other hull path from $u$ to $v$.}
\label{fig:rectangle-one}
\end{figure}

\begin{lem}\label{lem:rectangle}
Let $u,v\in P$, let $L=|uv|$, and suppose that
$\angle u>2\pi/3$ and $\angle v>2\pi/3$. Then one of the following holds:
\begin{enumerate}[label=(\arabic*)]
    \item one of the two hull paths from $u$ to $v$ has length at most
    $\sqrt{290}\,L$; or
    \item there exists a rectangle $Q_2$, centered at the midpoint of $uv$, with
    dimensions $6L\times\frac{24}{17}L$, such that every point of $\ch(P)$ whose
    orthogonal projection onto the long axis of $Q_2$ lies within distance $3L$ of
    the midpoint of $uv$ is contained in $Q_2$.
\end{enumerate}
\end{lem}

\begin{proof}
If $u$ and $v$ are adjacent on $\ch(P)$, then the hull path consisting of
the single edge $uv$ has length $L\le \sqrt{290}\,L$, so outcome~(1) holds.
Hence, for the remainder of the proof, we may assume that $u$ and $v$ are
nonadjacent.
Set $\tau=2\arctan(1/17) \approx 0.0374\,\pi$. Then $\tan(\tau/2)=1/17$ and
\[
\frac{1}{\sin(\tau/2)}=\sqrt{290}.
\]

If one of the two hull paths $H_1,H_2$ satisfies $\phi(H_i)\ge\tau$, then
Observation~\ref{obs:path-in-triangle} implies that one of the two hull paths, say
$H_j$, is contained in the triangle $x(H_i)uv$. Hence, by
Observations~\ref{obs:convex-chain} and~\ref{obs:triangle},
\[
\ell(H_j)
\le |ux(H_i)|+|x(H_i)v|
\le \frac{1}{\sin(\phi(H_i)/2)}\,|uv|
\le \frac{1}{\sin(\tau/2)}\,L
=
\sqrt{290}\,L.
\]
Thus outcome~(1) holds.

We may therefore assume that $\phi(H_1)<\tau$ and $\phi(H_2)<\tau$. Choose one of the two hull paths, say $H=H_1$. If $\lambda_u(H)$ and
$\lambda_v(H)$ are parallel, let $W$ be the closed strip bounded by these
lines and let $d$ be their common direction. Otherwise, let $W$ be the
closed wedge bounded by these lines that contains $\ch(P)$, and let $d$ be
a bisector direction of this wedge. Apply a rigid motion so that the
midpoint $m$ of $uv$ is the origin and the direction $d$ is horizontal, as
shown in Figure~\ref{fig:rectangle-two}.

\begin{figure}
\begin{center}
\includegraphics{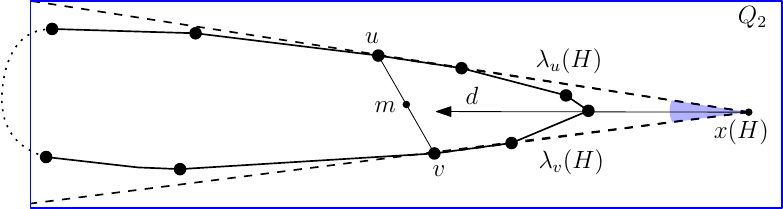}
\end{center}
\caption{$Q_2$ is centered at $m$ and is axis-parallel in the coordinate
system where the direction $d$ is horizontal. The dimensions of $Q_2$
are chosen so that $\lambda_u(H)$ and $\lambda_v(H)$ do not cross the
long sides of $Q_2$. Since $\ch(P)\subseteq W$, the convex hull can exit
$Q_2$ only through its short sides.}
\label{fig:rectangle-two}
\end{figure}

Let $\rho=\tan(\phi(H)/2)$. Since $\phi(H)<\tau$, we have $\rho<1/17$. In this coordinate
system the two boundary lines of $W$ have equations $y=\rho x+b_u$ and $y=-\rho x+b_v$ for
some constants $b_u,b_v$. Since the midpoint of
$uv$ is the origin, $|x_u|,|y_u|,|x_v|,|y_v|\le L/2$. Therefore
\[
|b_u|=|y_u-\rho x_u|\le \frac{1+\rho}{2}L,
\qquad
|b_v|=|y_v+\rho x_v|\le \frac{1+\rho}{2}L.
\]

Now let $p=(x_p,y_p)$ be a point of $\ch(P)$ whose orthogonal projection onto the
horizontal axis lies within distance $3L$ of the origin. Then $|x_p|\le 3L$, and
because $\ch(P)\subseteq W$,
\[
|y_p|
\le
\max\{|\rho x_p+b_u|,\ |-\rho x_p+b_v|\}
\le
3\rho L+\frac{1+\rho}{2}L
<
\frac{3}{17}L+\frac{9}{17}L
=
\frac{12}{17}L.
\]
Thus every such point lies in the rectangle
\[
Q_2=[-3L,3L]\times\left[-\frac{12}{17}L,\frac{12}{17}L\right],
\]
after the above rigid motion. This proves outcome~(2).
\end{proof}

We will make use of the following length bound for monotone convex chains.

\begin{obs}\label{obs:short-subchains}
Let $\gamma$ be an $x$-monotone convex polygonal chain contained in the rectangle
$[-bL,bL]\times[-hL,hL]$, for $b,h,L\ge 0$, and suppose that $\gamma$ starts at a point
$z=(x_z,y_z)$. Then
\[
\ell(\gamma)\le bL+|x_z|+3hL+|y_z|.
\]
\end{obs}

\begin{proof}
Since $\gamma$ is $x$-monotone, its horizontal variation is at most $bL+|x_z|$.
Moreover, the $y$-coordinate along an $x$-monotone convex chain has at most one local
extremum. Since the chain is contained in the strip $-hL\le y\le hL$, its vertical
variation is at most $3hL+|y_z|$. Finally, the length of a polygonal chain is at most
the sum of its horizontal and vertical variations.
\end{proof}

\begin{lem}\label{lem:local-chains}
Assume that outcome~(2) of Lemma~\ref{lem:rectangle} holds, and let $Q_2$ be the
corresponding rectangle. Suppose that neither hull path from $u$ to $v$ has
length at most $\kappaGT L$. Then $\partial\ch(P)\cap Q_2$ consists of two
$x$-monotone convex chains $A$ and $B$, each meeting both short sides of $Q_2$,
where $u\in A$ and $v\in B$. Moreover,
$P\cap Q_2=V(A)\cup V(B)$, where $V(A)=A\cap P$ and $V(B)=B\cap P$.
\end{lem}

\begin{proof}
We work in the coordinate system introduced in the proof of Lemma~\ref{lem:rectangle}, so
\[
Q_2=[-3L,3L]\times\left[-\frac{12}{17}L,\frac{12}{17}L\right].
\]
By Lemma~\ref{lem:rectangle}, every point of $\ch(P)$ with $|x|\le 3L$ lies in $Q_2$.

We first claim that the convex hull meets both vertical lines $x=-3L$ and $x=3L$.
Indeed, suppose for example that $\ch(P)$ does not meet the line $x=-3L$. Then the
part of $\partial\ch(P)$ outside the strip $|x|\le 3L$, if nonempty, lies to the
right of the line $x=3L$ and forms a single boundary interval, by convexity. Since
$u$ and $v$ lie in the strip, at most one of the two hull paths from $u$ to $v$
contains this interval. The other hull path is contained entirely in the strip, and
therefore is contained in $Q_2$ by Lemma~\ref{lem:rectangle}. 
  Since it is a convex chain contained in a rectangle, its length would be at
most
\[
\operatorname{per}(Q_2)
=
2\left(6L+\frac{24}{17}L\right)
=
\frac{252}{17}L
<\kappaGT L,
\]
contrary to the assumption. The case where $\ch(P)$ does not meet $x=3L$ is symmetric.

Thus $\ch(P)$ meets both vertical sides of the strip. Since $\ch(P)$ is convex, the
portion of its boundary inside the strip consists of two $x$-monotone convex chains,
one upper and one lower, each meeting both vertical sides of the strip. Both chains are contained in $Q_2$, and therefore they meet both short
sides of $Q_2$.

Finally, $u$ and $v$ cannot lie on the same one of these two chains; otherwise the subchain between them would be a hull path contained in $Q_2$,
again of length less than $\kappaGT L$. We name the two chains $A$ and $B$ so that $u\in A$ and $v\in B$. Since
all points of $P$ lie on $\partial\ch(P)$, every point of $P\cap Q_2$ belongs to one
of these two chains, so $P\cap Q_2=V(A)\cup V(B)$.
\end{proof}
From now on, whenever outcome~(2) of Lemma~\ref{lem:rectangle} holds, let $Q_1$ be the
rectangle concentric with $Q_2$, with the same height as $Q_2$, and with dimensions
$\frac85L\times\frac{24}{17}L$. We work in the coordinate system in which the midpoint of $uv$ is the origin and
\[
Q_1=
\left[-\frac45L,\frac45L\right]
\times
\left[-\frac{12}{17}L,\frac{12}{17}L\right],
\qquad
Q_2=
[-3L,3L]
\times
\left[-\frac{12}{17}L,\frac{12}{17}L\right].
\]
By reflecting the point set in the vertical axis if necessary, we assume that
$x_u\le x_v$. Since the midpoint of $uv$ is the origin, this implies
$x_u\le 0\le x_v$.

Partition $Q_2$ into four regions:
\[
R_1=(Q_2\setminus Q_1)\cap\{x<0\},\qquad
R_2=Q_1\cap\{x\le 0\},
\]
\[
R_3=Q_1\cap\{x>0\},\qquad
R_4=(Q_2\setminus Q_1)\cap\{x>0\}.
\]

Refer to Figure~\ref{fig:configuration} for an illustration of these regions.

\begin{figure}
\begin{center}
\includegraphics{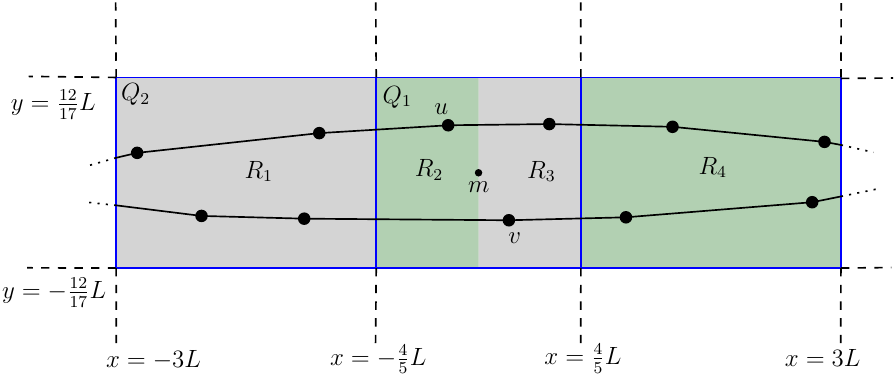}
\end{center}
\caption{The four regions $R_1,R_2,R_3,R_4$ that $Q_2$ is partitioned into. Note that $R_2,R_3$ partition the smaller rectangle $Q_1$.}
\label{fig:configuration}
\end{figure}

\begin{obs}\label{obs:metric-bounds}
In the above setting, the following hold.
\begin{enumerate}[label=(\arabic*)]
    \item $\operatorname{diam}(Q_1)<\frac{11}{5}L$.

    \item $\operatorname{diam}(R_1\cup R_2)=\operatorname{diam}(R_3\cup R_4)<\frac{19}{5}L$.

    \item If $p\in R_1$ and $q\in Q_2$ has $x_q\le x_v$, then
    $|pq|<\frac{19}{5}L$. The symmetric statement holds if $p\in R_4$ and
    $q\in Q_2$ has $x_q\ge x_u$.

    \item If $p\in R_4$ and $q\in Q_2\setminus R_1$, then $|pq|<6L$.
    The symmetric statement holds if $p\in R_1$ and $q\in Q_2\setminus R_4$.

    \item If $p\in Q_1$ and $q$ lies strictly to the left or to the right of $Q_2$,
    then $|pq|>\frac{11}{5}L$.

    \item If $p$ lies strictly to the left of $Q_2$ and $q\in R_4$, then
    $|pq|>\frac{19}{5}L$. The symmetric statement holds if $p$ lies strictly to the
    right of $Q_2$ and $q\in R_1$.

    \item If $p$ lies strictly to the left of $Q_2$ and $q$ lies strictly to the right
    of $Q_2$, then $|pq|>6L$.
\end{enumerate}
\end{obs}

\begin{proof}
The first bound follows from
\[
\operatorname{diam}(Q_1)
=
2L\sqrt{\left(\frac45\right)^2+\left(\frac{12}{17}\right)^2}
<
\frac{11}{5}L.
\]
For (2), the rectangle $R_1\cup R_2$ has width $3L$ and height
$\frac{24}{17}L$, so its diameter is
\[
L\sqrt{9+\left(\frac{24}{17}\right)^2}
<
\frac{19}{5}L.
\]
The same holds for $R_3\cup R_4$.

For (3), if $p\in R_1$, then $x_p\ge -3L$, while $x_q\le x_v\le L/2$. Thus the
horizontal distance between $p$ and $q$ is at most $\frac72L$, and the vertical
distance is at most $\frac{24}{17}L$. Hence
\[
|pq|
\le
L\sqrt{\left(\frac72\right)^2+\left(\frac{24}{17}\right)^2}
<
\frac{19}{5}L.
\]
The symmetric statement is identical.

For (4), the maximum possible horizontal distance between a point of $R_4$ and a
point of $Q_2\setminus R_1$ is $\frac{19}{5}L$, and the maximum vertical distance is
$\frac{24}{17}L$, so the distance is less than
\[
L\sqrt{\left(\frac{19}{5}\right)^2+\left(\frac{24}{17}\right)^2}<\frac{21}{5}L<6L.
\]

Finally, (5), (6), and (7) follow directly from horizontal separation. The horizontal
distance from $Q_1$ to either vertical exterior of $Q_2$ is greater than
$3L-\frac45L=\frac{11}{5}L$; the horizontal distance from the left exterior of $Q_2$
to $R_4$ is greater than $3L+\frac45L=\frac{19}{5}L$; and the horizontal distance
from the left exterior of $Q_2$ to the right exterior of $Q_2$ is greater than $6L$.
\end{proof}

For the rest of the proof, fix $u,v\in P$, set $L=|uv|$, and assume that
$\angle u>2\pi/3$ and $\angle v>2\pi/3$. Assume also that neither hull path
from $u$ to $v$ has length at most $\kappaGT L$. Since
$\sqrt{290}<\kappaGT$, outcome~(2) of Lemma~\ref{lem:rectangle} holds.
By Lemma~\ref{lem:local-chains}, the intersection
$\partial\ch(P)\cap Q_2$ consists of two $x$-monotone convex chains $A$ and
$B$, each joining the two short sides of $Q_2$, with $u\in A$ and $v\in B$.
Write $V(A)=A\cap P$ and $V(B)=B\cap P$. This configuration is shown in
Figure~\ref{fig:configuration}.

Call an edge $ab\in E(\T)$ a \emph{bridge} if
\[
 a\in V(A),\qquad b\in V(B),\qquad a,b\in Q_2.
\]
The next lemma explains why a bridge is exactly what we need for the spanning-ratio bound.

\begin{lem}\label{lem:bridge-short-path}
Let $u,v\in P$, let $L=|uv|$, and work in a coordinate system in which the midpoint
of $uv$ is the origin. Let
\[
Q_2=[-3L,3L]\times\left[-\frac{12}{17}L,\frac{12}{17}L\right].
\]
Suppose that $\partial\ch(P)\cap Q_2$ consists of two $x$-monotone convex chains
$A$ and $B$, with $u\in A$ and $v\in B$, and that
$P\cap Q_2=V(A)\cup V(B)$, where $V(A)=A\cap P$ and $V(B)=B\cap P$.
If $\T$ contains an edge $ab$ with
\[
 a\in V(A),\qquad b\in V(B),\qquad a,b\in Q_2,
\]
then $d_\T(u,v)\le \kappaGT L$.
\end{lem}

\begin{proof}
Let $ab$ be a bridge with $a\in V(A)$ and $b\in V(B)$, as shown in Figure~\ref{fig:bridge-path}(ii).
Let $\pi_A$ be the subpath of $A$ from $u$ to $a$, and let $\pi_B$ be the subpath of $B$ from $b$ to $v$.
Then $\pi_A$, followed by the edge $ab$, followed by $\pi_B$, is a path in $\T$ from $u$ to $v$.

Apply Observation~\ref{obs:short-subchains} to $\pi_A$ and $\pi_B$ with rectangle parameters $b=3$ and $h=12/17$.
Since the midpoint of $uv$ is the origin, the coordinates satisfy $v=-u$, and therefore
\[
|x_u|+|x_v|+|y_u|+|y_v|=2(|x_u|+|y_u|)\le \sqrt{2}\,L.
\]
Thus
\[
\ell(\pi_A)+\ell(\pi_B)
\le 6L+6\cdot\frac{12}{17}L+\sqrt2\,L
=\left(6+\frac{72}{17}+\sqrt2\right)L.
\]
Moreover,
\[
|ab|\le \diam(Q_2)=2L\sqrt{9+\left(\frac{12}{17}\right)^2}.
\]
Consequently,
\[
\begin{aligned}
d_\T(u,v)
&\le \ell(\pi_A)+|ab|+\ell(\pi_B)\\
&\le
\left(
6+\frac{72}{17}+\sqrt2
+2\sqrt{9+\frac{144}{289}}
\right)L\\
&=\kappaGT L
<18L.
\end{aligned}
\]
\end{proof}

For the next three lemmas we argue under the contradiction assumption that no bridge
exists; this assumption is included in each statement below.

\begin{lem}\label{lem:R1-nonempty}
Let $u,v\in P$, let $L=|uv|$, and suppose that outcome~(2) of
Lemma~\ref{lem:rectangle} holds with the rectangles $Q_1,Q_2$ and regions
$R_1,R_2,R_3,R_4$ defined above. Suppose that neither hull path from $u$ to $v$
has length at most $\kappaGT L$, and let $A$ and $B$ be the two chains given by
Lemma~\ref{lem:local-chains}. If no bridge exists in $Q_2$, then
$R_1\cap P\neq\emptyset$.
\end{lem}

\begin{figure}
\begin{center}
\includegraphics{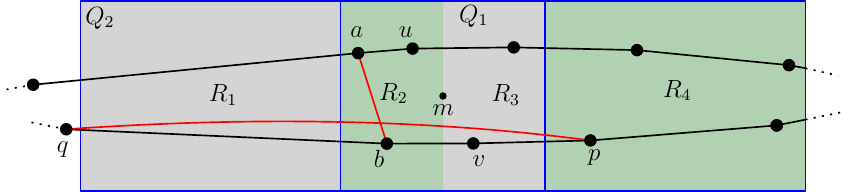}
\end{center}
\caption{If $R_1$ does not contain any points, any blocker of $ab$ must have one endpoint $q$ to the left of $Q_2$, and the other endpoint $p$ in $Q_1$, in $R_4$, or to the right of $Q_2$. In each case, $pq$ is shown to be longer than $ab$. Here, $pq$ has been drawn slightly curved for illustration purposes. }
\label{fig:nonempty}
\end{figure}

\begin{proof}
Refer to Figure~\ref{fig:nonempty}.
Assume for contradiction that $R_1\cap P=\emptyset$. Since $u\in V(A)\cap Q_1$ and
$v\in V(B)\cap Q_1$, the sets $V(A)\cap Q_1$ and $V(B)\cap Q_1$ are nonempty.
Let $a$ be the leftmost vertex of $V(A)\cap Q_1$, and let $b$ be the leftmost
vertex of $V(B)\cap Q_1$.

By Observation~\ref{obs:metric-bounds}(1), $|ab|<\frac{11}{5}L$. Since no edge of
$\T$ joins $V(A)$ to $V(B)$ inside $Q_2$, the segment $ab$ is not an edge of $\T$.
By Observation~\ref{obs:blocker}, there is a blocker $pq\in E(\T)$ crossing
$ab$. By the choice of $ab$ as leftmost, one endpoint of $pq$ must lie strictly to the
left of $Q_2$; call this endpoint $q$.

The other endpoint $p$ cannot lie strictly to the left of $Q_2$ or in $R_1$: in
either case, the segment $pq$ is disjoint from the relative interior of
$ab\subseteq Q_1$, and so it cannot block $ab$. Thus $p$ lies in $Q_1$, in $R_4$,
or strictly to the right of $Q_2$. In these three cases,
Observations~\ref{obs:metric-bounds}(5), (6), and (7) give respectively
$|pq|>\frac{11}{5}L$, $|pq|>\frac{19}{5}L$, and $|pq|>6L$. In every case,
$|pq|>\frac{11}{5}L>|ab|$, contradicting the fact that $pq$ is a blocker of
$ab$. Therefore $R_1\cap P\neq\emptyset$.
\end{proof}

Since $P\cap Q_2=V(A)\cup V(B)$ and $R_1\cap P\neq\emptyset$, either
$R_1\cap V(A)\neq\emptyset$ or $R_1\cap V(B)\neq\emptyset$. These two cases are
symmetric after exchanging the roles of $A$ and $B$, so assume without loss of
generality that $R_1\cap V(A)\neq\emptyset$.

We now describe a process that we call the \emph{left-to-right sweep}. This process is shown in Figure~\ref{fig:left-to-right-one}(i). Let $a_1$ be the leftmost vertex of $V(A)$.
Since $R_1\cap V(A)\neq\emptyset$, we have $a_1\in R_1$. Among all segments $a_1b$
with $b\in V(B)$, choose a leftmost one and denote it by $a_1b_1$. Equivalently,
$b_1$ is the leftmost vertex of $V(B)$. Since $v\in V(B)$, we have $x_{b_1}\le x_v$.

\begin{figure}
\begin{center}
\includegraphics{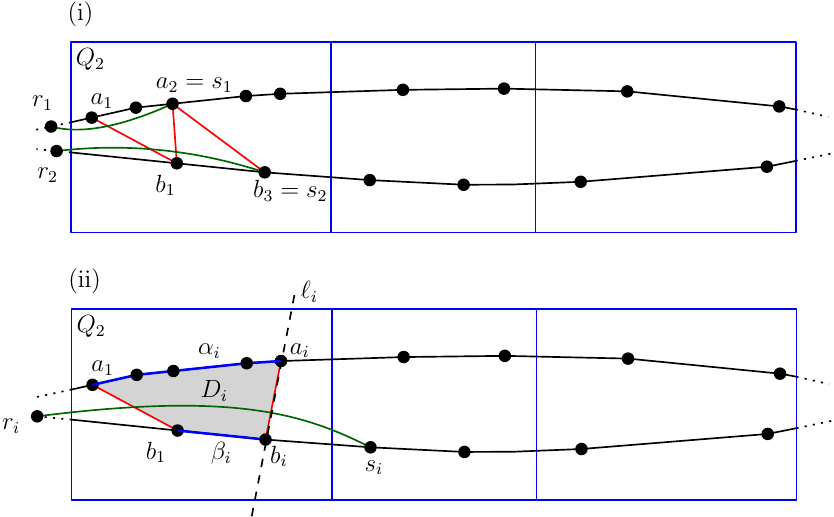}
\end{center}
\caption{The left-to-right sweep.
(i) The red segments are successive segments of the form $a_ib_i$, and the green edges are their blockers. In the first step, the blocker $e_1$ advances the upper endpoint, so $a_2=s_1$ and $b_2=b_1$; in the second step, the blocker $e_2$ advances the lower endpoint, so $a_3=a_2$ and $b_3=s_2$.
(ii) At step $i$, the paths $\alpha_i\subseteq A$ and $\beta_i\subseteq B$, together with $a_1b_1$ and $a_ib_i$, bound the region $D_i$. The line $\ell_i$ contains $a_ib_i$, and $H_i$ is the side of $\ell_i$ containing $D_i$. The blocker is $e_i=r_is_i$, where $r_i$ is the endpoint lying in $H_i$ and $s_i$ is the other endpoint.}
\label{fig:left-to-right-one}
\end{figure}

Starting from $a_1b_1$, recursively construct a sequence of segments
$a_1b_1,a_2b_2,\ldots,a_mb_m$ as follows. Suppose that $a_ib_i$ has been defined.
Since no edge of $\T$ joins $V(A)$ to $V(B)$ inside $Q_2$, the segment $a_ib_i$ is
not an edge of $\T$. Let $e_i\in E(\T)$ be a blocker of $a_ib_i$.

For $i=1$, the choice of $a_1b_1$ implies that $e_1$ has an endpoint strictly to
the left of $Q_2$; call that endpoint $r_1$, and call the other endpoint $s_1$.
For $i\ge2$, let $\alpha_i$ be the subpath of $A$ from $a_1$ to $a_i$, and let
$\beta_i$ be the subpath of $B$ from $b_1$ to $b_i$. Let $D_i$ be the closed region bounded by
$a_1b_1$, $\alpha_i$, $a_ib_i$, and $\beta_i$. Let $\ell_i$ be the line containing
$a_ib_i$, and let $H_i$ be the open halfplane bounded by $\ell_i$ that contains
interior points of $D_i$. Denote by $r_i$ the endpoint of $e_i$ that lies in
$H_i$, and by $s_i$ the other endpoint. 

If $s_i\in V(A)$, set $a_{i+1}=s_i$ and $b_{i+1}=b_i$. If $s_i\in V(B)$, set
$a_{i+1}=a_i$ and $b_{i+1}=s_i$. If $s_i\notin V(A)\cup V(B)$, the left-to-right sweep terminates.
As shown in the next lemma, whenever the sweep continues, one of the two endpoints
moves strictly forward along its chain. Since $P$ is finite, the process terminates
after finitely many steps. Let $a_mb_m$ be the last
segment obtained.

\begin{lem}\label{lem:left-sweep-properties}
Assume that no bridge exists in $Q_2$ and, after possibly exchanging $A$ and $B$,
that $R_1\cap V(A)\neq\emptyset$. For the left-to-right sweep defined above, the
following hold.
\begin{enumerate}[label=(\roman*)]
    \item For every $i\in\{1,\ldots,m\}$, the point $r_i$ lies strictly to the left of
    $Q_2$.
    \item There exists an index $k\in\{1,\ldots,m\}$ such that
    $s_k\in Q_1\cap(V(A)\cup V(B))$.
\end{enumerate}
\end{lem}

\begin{proof}
We first prove (i). For $i=1$, the claim follows from the choice of $a_1b_1$ as a
leftmost segment from $V(A)$ to $V(B)$. If both endpoints of a blocker of $a_1b_1$
lie in $Q_2$, then the crossing with $a_1b_1$ would produce a segment from $V(A)$ to
$V(B)$ lying strictly to the left of $a_1b_1$, contradicting the choice of $a_1b_1$.

Now let $i\ge 2$, and assume the claim holds for all smaller indices. Every point of
$P\cap Q_2\cap H_i$ belongs to $D_i$: indeed, every point of $P\cap Q_2$ lies on
$A$ or on $B$, and the points of $A\cup B$ lying in $H_i$ are exactly on the already
traversed subpaths $\alpha_i$ and $\beta_i$.

We claim that neither endpoint of $e_i$ lies in $D_i$. Figure~\ref{fig:left-to-right-one}(ii) illustrates why this should be true. Suppose for contradiction
that one endpoint of $e_i$ lies in $D_i$. Then this endpoint lies on
$\alpha_i\cup\beta_i$. If it lies on $\beta_i$, choose the smallest index
$j\in\{2,\ldots,i\}$ such that $x_{b_j}>x_{r_i}$. Then
$x_{b_{j-1}}\le x_{r_i}<x_{b_j}$, so $b_j\neq b_{j-1}$. Hence at step $j-1$ the
$B$-endpoint advanced, and the blocker edge $e_{j-1}=r_{j-1}s_{j-1}$ satisfies
$s_{j-1}=b_j$. The edge $e_i$ starts on $\beta_i$ strictly to the left of $b_j$ and
crosses $a_ib_i$, while $e_{j-1}$ enters $Q_2$ from the left and ends at $b_j$.
These two edges cross, contradicting planarity of $\T$.

The same argument applies if the endpoint of $e_i$ lies on $\alpha_i$: choose the
smallest $j$ such that $x_{a_j}>x_{r_i}$. Then the edge that introduced $a_j$ crosses
$e_i$, again contradicting planarity. Thus neither endpoint of $e_i$ lies in
$D_i$.

Since $\alpha_i$ and $\beta_i$ are contained in $D_i$, it follows that
$s_i\notin V(\alpha_i)\cup V(\beta_i)$. Hence, if $s_i\in V(A)$, then $s_i$ lies on
the suffix of $A$ strictly after $a_i$, and so $x_{s_i}>x_{a_i}$. Similarly, if
$s_i\in V(B)$, then $s_i$ lies on the suffix of $B$ strictly after $b_i$, and
$x_{s_i}>x_{b_i}$.

If $r_i$ were contained in $Q_2$, then
$r_i\in P\cap Q_2\cap H_i$, and therefore $r_i\in D_i$, contradicting the
previous paragraph. Hence $r_i\notin Q_2$.

The line $\ell_i$ separates the two boundary paths of $\partial\ch(P)$ between
$a_i$ and $b_i$. By the definition of $H_i$, the boundary path contained in
$\overline{H_i}$ is the one passing through the left side of $Q_2$, whereas
the other boundary path passes through the right side. Therefore every vertex
of $P\cap H_i$ outside $Q_2$ lies strictly to the left of $Q_2$. Since
$r_i\in H_i$, the point $r_i$ lies strictly to the left of $Q_2$. This proves
(i).

We now prove (ii). The initial segment satisfies $a_1\in R_1$ and $x_{b_1}\le x_v$.
Let $k$ be the largest index such that $a_k\in R_1$ and $x_{b_k}\le x_v$. By (i),
$r_k$ lies strictly to the left of $Q_2$. We claim that
$s_k\in Q_1\cap(V(A)\cup V(B))$. An example is shown in Figure~\ref{fig:left-to-right-two}, where the largest such index is $k=4$.

\begin{figure}
\begin{center}
\includegraphics{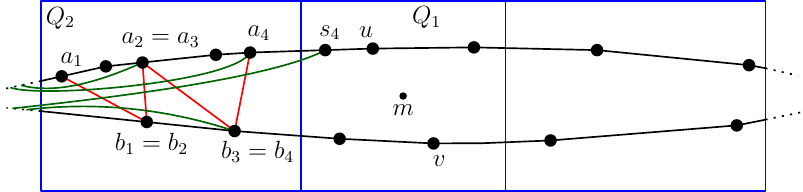}
\end{center}
\caption{If $a_4b_4$ is the last segment for which $a_i\in R_1$ and
$x_{b_i}\le x_v$, then $s_4\in Q_1$.}
\label{fig:left-to-right-two}
\end{figure}

Suppose not. The point $s_k$ cannot lie strictly to the left of $Q_2$, since $r_k$
also lies strictly to the left of $Q_2$, in which case the segment
$r_ks_k$ could not cross $a_kb_k\subseteq Q_2$. If $s_k\in R_1$, then the next segment would still satisfy
$a_{k+1}\in R_1$ and $x_{b_{k+1}}\le x_v$, contradicting the maximality of $k$. If
$s_k\in R_4$, then Observation~\ref{obs:metric-bounds}(6) gives
$|r_ks_k|>\frac{19}{5}L$, while Observation~\ref{obs:metric-bounds}(3) gives
$|a_kb_k|<\frac{19}{5}L$, contradicting that $r_ks_k$ is a blocker of
$a_kb_k$. If $s_k$ lies strictly to the right of $Q_2$, then
Observation~\ref{obs:metric-bounds}(7) gives $|r_ks_k|>6L$, again contradicting
$|a_kb_k|<\frac{19}{5}L$. Therefore
$s_k\in Q_1\cap(V(A)\cup V(B))$.
\end{proof}

In the previous lemma, we identified an edge $r_ks_k$ with one endpoint outside
$Q_2$ and the other endpoint in $Q_1$. The purpose of the next lemma is to
identify a second such edge. The key additional property is that its endpoint
inside $Q_1$ lies on the opposite hull chain from the endpoint of the first edge.

\begin{lem}\label{lem:right-sweep}
Assume that no bridge exists in $Q_2$. Let $a_1b_1,\ldots,a_mb_m$ be the sequence produced by the left-to-right sweep, and
let $k$ be the index given by Lemma~\ref{lem:left-sweep-properties}. Then the following
hold.
\begin{enumerate}[label=(\roman*)]
    \item If $s_k\in V(A)$, then $a_m\in R_4$ and $b_m\in R_1$, and there exists an
    edge $pq\in E(\T)$ such that $p$ lies strictly to the right of $Q_2$ and
    $q\in Q_1\cap V(B)$.
    \item If $s_k\in V(B)$, then $a_m\in R_1$ and $b_m\in R_4$, and there exists an
    edge $pq\in E(\T)$ such that $p$ lies strictly to the right of $Q_2$ and
    $q\in Q_1\cap V(A)$.
\end{enumerate}
\end{lem}

\begin{proof}
Since the sweep terminates at $a_mb_m$, we have $s_m\notin V(A)\cup V(B)$. By
Lemma~\ref{lem:left-sweep-properties}(i), the point $r_m$ lies strictly to the left of
$Q_2$. The point $s_m$ cannot also lie strictly to the left of $Q_2$, since then the
edge $r_ms_m$ would lie entirely to the left of $Q_2$ and could not cross
$a_mb_m\subseteq Q_2$. Since $P\cap Q_2=V(A)\cup V(B)$, the point $s_m$ lies strictly
to the right of $Q_2$.

Thus Observation~\ref{obs:metric-bounds}(7) gives $|r_ms_m|>6L$. Since $r_ms_m$ is a
blocker of $a_mb_m$, we have $|a_mb_m|>6L$. If one endpoint of $a_mb_m$ lay
in $R_4$ and the other in $Q_2\setminus R_1$, then
Observation~\ref{obs:metric-bounds}(4) would give $|a_mb_m|<6L$. The symmetric
statement gives the same contradiction if one endpoint lay in $R_1$ and the other in
$Q_2\setminus R_4$. Hence one endpoint of $a_mb_m$ lies in $R_1$ and the other lies
in $R_4$; see Figure~\ref{fig:terminal}.

\begin{figure}
\begin{center}
\includegraphics{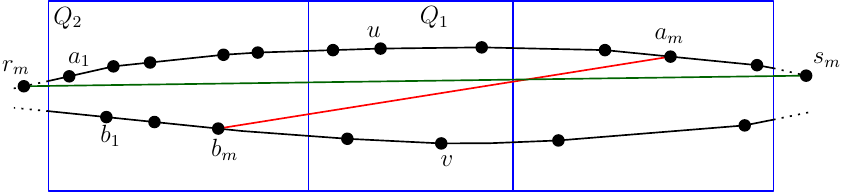}
\end{center}
\caption{The last blocker $r_ms_m$ must have its right endpoint to the right of $Q_2$; otherwise, if $s_m \in Q_2$, the process would continue. The length of $r_ms_m$ implies that $a_m\in R_4$, $b_m \in R_1$ or $a_m\in R_1$, $b_m\in R_4$.}
\label{fig:terminal}
\end{figure}

Assume first that $s_k\in V(A)$. The proof for the case $s_k\in V(B)$ is symmetric. Then
$s_k\in Q_1\cap V(A)$ by Lemma~\ref{lem:left-sweep-properties}(ii). From step $k$
onward, every later $A$-endpoint lies on the suffix of $A$ strictly after $s_k$.
Since $A$ is $x$-monotone and $s_k\in Q_1$, no later $A$-endpoint can belong to
$R_1$. Hence $a_m\notin R_1$, and therefore $a_m\in R_4$ and $b_m\in R_1$.

Let $z_1$ be the rightmost vertex of $V(B)$. Since $v\in V(B)$, we have
$x_{z_1}\ge x_v$, and hence $z_1\in Q_1\cup R_4$.

Let $j_0$ be the smallest index such that $a_{j_0}=a_m$. Then $j_0\ge 2$,
and the blocker edge of the preceding segment satisfies
$s_{j_0-1}=a_m$. By Lemma~\ref{lem:left-sweep-properties}(i), the other
endpoint $r_{j_0-1}$ lies strictly to the left of $Q_2$. Thus
$r_{j_0-1}a_m$ is an edge of $\T$ entering $Q_2$ from the left and ending
at $a_m$.

First suppose that $z_1\in Q_1$, as shown in Figure~\ref{fig:first-case}.
Since $a_mz_1$ is not an edge of $\T$, let $pq\in E(\T)$ be a blocker of
$a_mz_1$.

We claim that $pq$ has an endpoint strictly to the right of $Q_2$. The following argument is shown in Figure~\ref{fig:first-case}(ii). Suppose that both endpoints
of $pq$ lie in $Q_2$. Then, they cannot lie on different chains, as then $pq$ would be a bridge. Therefore, $p,q \in V(A)$ or $p,q\in V(B)$. Because $pq$ crosses $a_mz_1$, one endpoint of $pq$ would
lie to the right of the segment $a_mz_1$. If it lies on $V(B)$, then it
lies strictly to the right of $z_1$, contradicting the choice of $z_1$. If it lies on
$V(A)$, then it lies on the suffix of $A$ strictly after $a_m$, and the edge $pq$
crosses the edge $r_{j_0-1}a_m$, contradicting planarity. Thus $pq$ cannot have both
endpoints in $Q_2$.

If one endpoint of $pq$ lies strictly to the left of $Q_2$, then the other endpoint
cannot lie on $V(B)$: indeed, since $pq$ crosses $a_mz_1$, such an endpoint would
have to lie strictly to the right of $z_1$ along $B$, again contradicting the choice of
$z_1$ as the rightmost vertex of $V(B)$. If it lies
strictly to the right of $Q_2$, then Observation~\ref{obs:metric-bounds}(7) gives
$|pq|>6L$, while Observation~\ref{obs:metric-bounds}(3), in its symmetric form, gives
$|a_mz_1|<\frac{19}{5}L$, a contradiction. If the other endpoint lies on $V(A)$, then
it lies in $R_4$, and Observation~\ref{obs:metric-bounds}(6) gives
$|pq|>\frac{19}{5}L$, again contradicting
$|a_mz_1|<\frac{19}{5}L$. Therefore $pq$ has an endpoint strictly to the right of
$Q_2$.

\begin{figure}
\begin{center}
\includegraphics{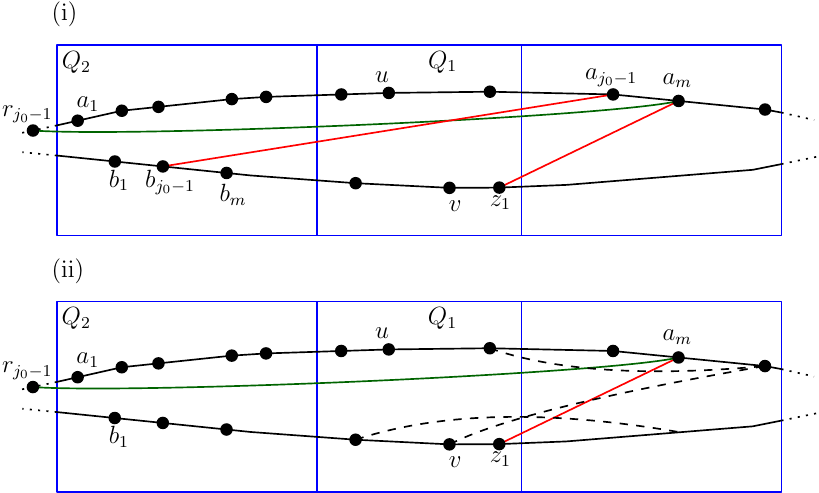}
\end{center}
\caption{(i) The case $z_1\in Q_1$. There exists a blocker $r_{j_0-1}a_m$, shown in green. (ii) The blocker of $a_mz_1$ cannot have both endpoints in $Q_2$, as shown by the three dashed segments. Such a blocker would have to either be a bridge, or violate planarity by crossing $r_{j_0-1}a_m$ or have an endpoint in $V(B)$ to the right of $z_1$.}
\label{fig:first-case}
\end{figure}

Let this right endpoint be $p$, as illustrated in Figure~\ref{fig:first-case-two}. The other endpoint cannot lie on $V(A)$, since then
$pq$ would cross the edge $r_{j_0-1}a_m$. Hence the other endpoint lies on $V(B)$.
Since $z_1$ is the rightmost vertex of $V(B)$ and $z_1\in Q_1$, this other endpoint
lies in $Q_1\cup R_1$. It cannot lie in $R_1$, because then
Observation~\ref{obs:metric-bounds}(6) gives $|pq|>\frac{19}{5}L$, while
$|a_mz_1|<\frac{19}{5}L$. Thus the other endpoint lies in $Q_1\cap V(B)$, proving
the desired statement in the case $z_1\in Q_1$.

\begin{figure}
\begin{center}
\includegraphics{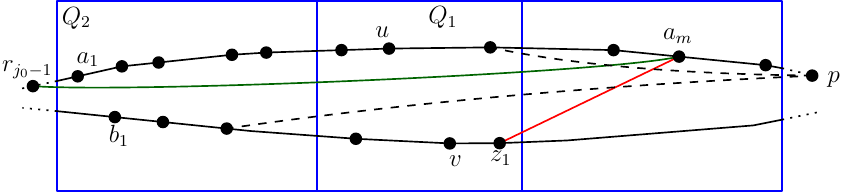}
\end{center}
\caption{The blocker from $p$ must have its other endpoint in $Q_1\cap V(B)$, as indicated by the two dashed segments. If its other endpoint is in $V(A)$, it would cross $r_{j_0-1}a_m$ violating planarity, and if it is not in $Q_1$, it would be too long compared to $a_mz_1$.}
\label{fig:first-case-two}
\end{figure}

It remains to consider the case $z_1\in R_4$. This case is illustrated in
Figure~\ref{fig:second-case}. We perform a right-to-left sweep along $B$,
keeping $a_m$ fixed. Starting from $a_mz_1$, recursively construct segments
$a_mz_1,a_mz_2,\ldots,a_mz_t$ as follows. Suppose that $a_mz_i$ has been
defined, with $z_i\in R_4$. Since $a_mz_i$ is not an edge of $\T$, let
$f_i\in E(\T)$ be a blocker of $a_mz_i$. Let $\gamma_i$ be the subpath of
$B$ from $z_i$ to $z_1$, and let $\Omega_i$ be the closed region bounded by
$a_mz_1$, $\gamma_i$, and $a_mz_i$. Let $p_i$ be the endpoint of $f_i$ on
the same side of the line through $a_mz_i$ as $\Omega_i$, and let $q_i$ be
the other endpoint.

\begin{figure}
\begin{center}
\includegraphics{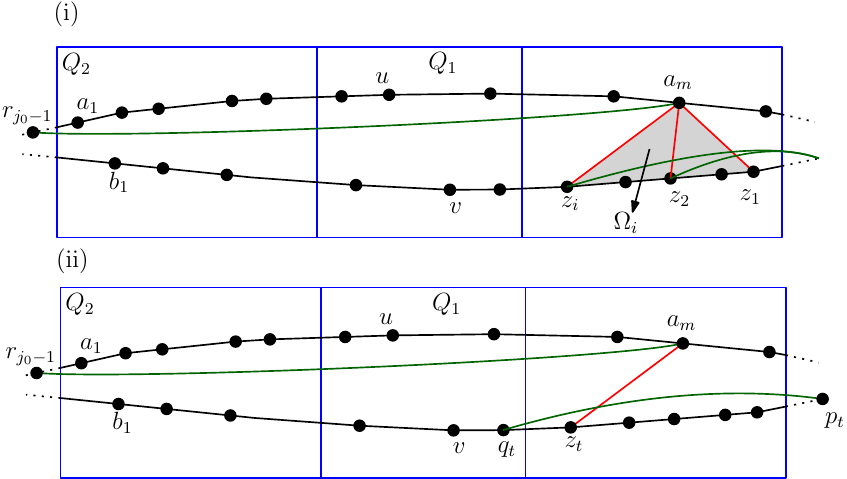}
\end{center}
\caption{The case $z_1\in R_4$. (i) The blockers of $a_mz_i$ must have an endpoint at the right of $Q_2$, while the other endpoint lies in $V(B)$. (ii) For the terminal segment $a_mz_t$ in the sweep, it must hold that $q_t\in Q_1\cap V(B)$. Note that $p_tq_t$ is the blocker of $a_mz_t$.}
\label{fig:second-case}
\end{figure}

The same argument used in Lemma~\ref{lem:left-sweep-properties}, reflected
horizontally, shows that $p_i$ lies strictly to the right of $Q_2$. Moreover,
if $q_i\in V(B)$, then $q_i$ lies on the suffix of $B$ strictly before
$z_i$, and hence
\[
x_{q_i}<x_{z_i}.
\]

If $q_i\in V(B)\cap R_4$, set $z_{i+1}=q_i$ and continue. Since the
$x$-coordinates strictly decrease and $P$ is finite, this process terminates.
Let $a_mz_t$ be the last segment obtained. Its blocker $p_tq_t$ satisfies
that $p_t$ lies strictly to the right of $Q_2$.

We claim that $q_t\in Q_1\cap V(B)$. If $q_t\in R_1$, then
Observation~\ref{obs:metric-bounds}(6) gives
$|p_tq_t|>\frac{19}{5}L$, whereas
Observation~\ref{obs:metric-bounds}(2) gives
$|a_mz_t|<\frac{19}{5}L$, a contradiction. If $q_t$ lies strictly to the
left of $Q_2$, then Observation~\ref{obs:metric-bounds}(7) gives
$|p_tq_t|>6L$, again contradicting
$|a_mz_t|<\frac{19}{5}L$. If $q_t\in R_4\cap V(B)$, then the sweep could
continue, contrary to the choice of $t$. Finally, $q_t$ cannot lie on
$V(A)$, because then $p_tq_t$ would cross the edge
$r_{j_0-1}a_m$ that introduced $a_m$. Thus
$q_t\in Q_1\cap V(B)$, and since $p_t$ lies to the right of $Q_2$, this
proves (i).

The proof of (ii) is symmetric, exchanging the roles of $A$ and $B$.
\end{proof}

We can now prove that a bridge must exist.

\begin{lem}\label{lem:bridge-exists}
Under the assumptions of Lemma~\ref{lem:local-chains}, $\T$ contains a bridge.
\end{lem}

\begin{proof}
Assume for contradiction that no bridge exists.
Then the sweep lemmas above apply.
Perform the left-to-right sweep, and let $k$ be the index from Lemma~\ref{lem:left-sweep-properties}.
Thus $s_k\in Q_1\cap(V(A)\cup V(B))$.

If $s_k\in V(A)$, set $x=s_k$.
By Lemma~\ref{lem:right-sweep}(i), there is an edge of $\T$ with one endpoint strictly to the right of $Q_2$ and the other endpoint $y\in Q_1\cap V(B)$.
If $s_k\in V(B)$, set $y=s_k$.
By Lemma~\ref{lem:right-sweep}(ii), there is an edge of $\T$ with one endpoint strictly to the right of $Q_2$ and the other endpoint $x\in Q_1\cap V(A)$.

In either case, we obtain points
\[
 x\in Q_1\cap V(A),\qquad y\in Q_1\cap V(B),
\]
together with an edge of $\T$ entering $Q_1$ from the left of $Q_2$ and ending at one of $x,y$, and an edge of $\T$ entering $Q_1$ from the right of $Q_2$ and ending at the other. This is depicted in Figure~\ref{fig:contradiction}.

\begin{figure}
\begin{center}
\includegraphics{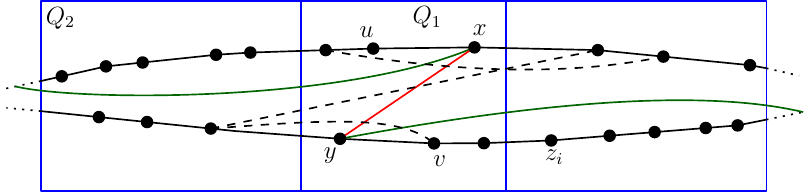}
\end{center}
\caption{Illustration for the proof of Lemma~\ref{lem:bridge-exists}. None of the three dashed segments can be a blocker of $xy$: such a segment would either be a bridge, or violate planarity by crossing one of the two green edges.}
\label{fig:contradiction}
\end{figure}

Consider the segment $xy$.
By Observation~\ref{obs:metric-bounds}(1), $|xy|<\frac{11}{5}L$.
If $xy\in E(\T)$, then $xy$ is a bridge, contrary to our assumption.
Hence $xy\notin E(\T)$, so by Observation~\ref{obs:blocker}, there is a blocker $e\in E(\T)$ crossing $xy$ with $|e|\le |xy|$.

We claim that $e$ cannot be contained entirely in $Q_2$.
Suppose otherwise.
Since $P\cap Q_2=V(A)\cup V(B)$, the endpoints of $e$ lie on $A\cup B$.
If $e$ has one endpoint in $V(A)$ and the other in $V(B)$, then $e$ is a bridge, again a contradiction.
Thus both endpoints of $e$ lie on the same chain.
Since $e$ crosses $xy$, these endpoints lie on opposite sides of $x$ along $A$, or on opposite sides of $y$ along $B$.
In the first case, $e$ crosses the already constructed edge of $\T$ incident to $x$; in the second case, it crosses the already constructed edge of $\T$ incident to $y$.
Both cases contradict planarity of $\T$.

Therefore $e$ has an endpoint outside $Q_2$. Since $e$ crosses the segment $xy$ at a point of $Q_1$, Observation~\ref{obs:metric-bounds}(5) shows that the subsegment from this exterior endpoint to the crossing point already has length greater than $\frac{11}{5}L$. Hence $|e|>\frac{11}{5}L$, contradicting $|e|\le |xy|<\frac{11}{5}L$.
This contradiction proves that a bridge exists.
\end{proof}

\begin{proof}[Proof of Theorem~\ref{thm:main}]
Fix $u,v\in P$, and set $L=|uv|$.

If $uv$ is a hull edge, then $uv\in E(\T)$, and hence
$d_\T(u,v)=L\le \kappaGT L$. Thus, assume that $u$ and $v$ are nonadjacent
vertices of $\ch(P)$.

By Corollary~\ref{cor:bounded-angle}, we may assume that
$\angle u>2\pi/3$ and $\angle v>2\pi/3$.

Apply Lemma~\ref{lem:rectangle}. If outcome~(1) holds, then one hull path
from $u$ to $v$ has length at most
\[
\sqrt{290}\,L<\kappaGT L,
\]
and we are done. Thus outcome~(2) holds.

If either hull path from $u$ to $v$ has length at most $\kappaGT L$, we are
again done. Hence we may assume that neither hull path has length at most
$\kappaGT L$.

By Lemma~\ref{lem:local-chains}, $\partial\ch(P)\cap Q_2$ consists of two
$x$-monotone convex chains $A$ and $B$, each joining the two short sides of
$Q_2$, with $u\in A$ and $v\in B$. Lemma~\ref{lem:bridge-exists} gives a
bridge between $A$ and $B$, and Lemma~\ref{lem:bridge-short-path} turns this
bridge into a $u$--$v$ path of length at most $\kappaGT L$.

Therefore,
\[
d_\T(u,v)\le \kappaGT |uv|.
\]
\end{proof}

\section{Conclusion}

For point sets in convex position, Theorem~\ref{thm:main} reduces the previously available
general-purpose constant by almost three orders of magnitude. A natural next question is to narrow the remaining gap between the lower bound of
$2.0268$ and the upper bound $\kappaGT<17.814$. It would also be interesting to determine whether the arguments presented can be extended to point sets with few interior points.

\bibliographystyle{plainurlnat}
\bibliography{GT}

@inproceedings{BoseLeeSmid2007,
  author    = {Prosenjit Bose and Aaron Lee and Michiel Smid},
  title     = {On Generalized Diamond Spanners},
  booktitle = {Algorithms and Data Structures: 10th International Workshop, WADS 2007},
  series    = {Lecture Notes in Computer Science},
  volume    = {4619},
  pages     = {325--336},
  publisher = {Springer},
  year      = {2007},
  doi       = {10.1007/978-3-540-73951-7_29}
}

@article{BoseSmid2013,
  author  = {Prosenjit Bose and Michiel Smid},
  title   = {On Plane Geometric Spanners: A Survey and Open Problems},
  journal = {Computational Geometry},
  volume  = {46},
  number  = {7},
  pages   = {818--830},
  year    = {2013},
  doi     = {10.1016/j.comgeo.2013.04.002}
}

@article{DumitrescuGhosh2016,
  author  = {Adrian Dumitrescu and Anirban Ghosh},
  title   = {Lower Bounds on the Dilation of Plane Spanners},
  journal = {International Journal of Computational Geometry \& Applications},
  volume  = {26},
  number  = {2},
  pages   = {89--110},
  year    = {2016},
  doi     = {10.1142/S0218195916500059},
  eprint  = {1509.07181},
  archivePrefix = {arXiv}
}

@article{LevcopoulosLingas1987,
  author  = {Christos Levcopoulos and Andrzej Lingas},
  title   = {On Approximation Behavior of the Greedy Triangulation for Convex Polygons},
  journal = {Algorithmica},
  volume  = {2},
  pages   = {175--193},
  year    = {1987},
  doi     = {10.1007/BF01840358}
}

@article{LevcopoulosLingas1992,
  author  = {Christos Levcopoulos and Andrzej Lingas},
  title   = {Fast Algorithms for Greedy Triangulation},
  journal = {BIT},
  volume  = {32},
  number  = {2},
  pages   = {280--296},
  year    = {1992},
  doi     = {10.1007/BF01994882}
}

@incollection{DasJoseph1989,
  author    = {Gautam Das and Deborah Joseph},
  title     = {Which Triangulations Approximate the Complete Graph?},
  booktitle = {Optimal Algorithms},
  series    = {Lecture Notes in Computer Science},
  volume    = {401},
  pages     = {168--192},
  publisher = {Springer},
  year      = {1989},
  doi       = {10.1007/3-540-51859-2_15}
}

@article{BuchinEtAl2026,
  author  = {Kevin Buchin and Joachim Gudmundsson and Antonia Kalb and Aleksandr Popov and Carolin Rehs and Andr{\'e} van Renssen and Sampson Wong},
  title   = {Oriented Spanners},
  journal = {Algorithmica},
  volume  = {88},
  year    = {2026},
  note    = {Article 14},
  doi     = {10.1007/s00453-025-01342-8}
}

\end{document}